\documentclass[a4paper,UKenglish]{lipics-v2018}

\usepackage{microtype}

\usepackage{amsmath, amsthm, amsfonts, amssymb}
\usepackage{breqn}
\usepackage{graphicx}
\usepackage{mathtools}
\usepackage{hyperref}
\usepackage[linesnumbered,ruled]{algorithm2e}
\usepackage{comment}
\usepackage{footmisc}
\usepackage{color}
\usepackage{wrapfig}

\newtheorem{claim}{Claim}

\newtheorem{proposition}{Proposition}

\DeclareMathOperator*{\argmin}{arg\,min}
\DeclareMathOperator*{\argmax}{arg\,max}

\newcommand{\eqdef}{\mathrel{\mathop:}=}

\DeclareMathOperator{\MRCE}{\mathcal{MRCE}}

\title{Maximum Rooted Connected Expansion}
\titlerunning{Maximum Rooted Connected Expansion} 

\author{Ioannis Lamprou}{Department of Computer Science, University of Liverpool, Liverpool, UK}{lamprou@liverpool.ac.uk}{}{}

\author{Russell Martin}{Department of Computer Science, University of Liverpool, Liverpool, UK}{ramartin@liverpool.ac.uk}{}{}

\author{Sven Schewe}{Department of Computer Science, University of Liverpool, Liverpool, UK}{svens@liverpool.ac.uk}{}{}

\author{Ioannis Sigalas}{Department of Informatics \& Telecommunications, University of Athens, Athens, Greece}{sigalasi@di.uoa.gr}{}{}

\author{Vassilis Zissimopoulos}{Department of Informatics \& Telecommunications, University of Athens, Athens, Greece}{vassilis@di.uoa.gr}{}{}

\authorrunning{I. Lamprou, R. Martin, S. Schewe, I. Sigalas and V. Zissimopoulos} 

\Copyright{Ioannis Lamprou, Russell Martin, Sven Schewe, Ioannis Sigalas and Vassilis Zissimopoulos}



\category{}

\relatedversion{}

\supplement{}

\funding{}

\acknowledgements{}

\EventEditors{Igor Potapov, Paul Spirakis, and James Worrell}
\EventNoEds{3}
\EventLongTitle{43rd International Symposium on Mathematical Foundations of Computer Science (MFCS 2018)}
\EventShortTitle{MFCS 2018}
\EventAcronym{MFCS}
\EventYear{2018}
\EventDate{August 27--31, 2018}
\EventLocation{Liverpool, GB}
\EventLogo{}
\SeriesVolume{117}
\ArticleNo{XX} 
\nolinenumbers 
\hideLIPIcs  

\begin{document}
	
	\maketitle
	
	\begin{abstract}
		\emph{Prefetching} constitutes a valuable tool toward the goal of efficient Web surfing.
		As a result, estimating the amount of resources that need to be preloaded during a surfer's browsing becomes an important task.
		In this regard, prefetching can be modeled as a two-player combinatorial game [Fomin et al.,\ \emph{Theoretical Computer Science 2014}], where a surfer and a marker alternately play on a given graph (representing the Web graph). During its turn, the marker chooses a set of $k$ nodes to mark (prefetch), whereas the surfer, represented as a token resting on graph nodes, moves to a neighboring node (Web resource). 
		The surfer's objective is to reach an unmarked node before all nodes become marked and the marker wins.
		Intuitively, since the surfer is step-by-step traversing a subset of nodes in the Web graph, a satisfactory prefetching procedure would load in cache (without any delay) all resources lying in the neighborhood of this growing subset. 
		
		Motivated by the above, we consider the following maximization problem to which we refer to as the \emph{Maximum Rooted Connected Expansion} (MRCE) problem. 
		Given a graph $G$ and a root node $v_0$, we wish to find a subset of vertices $S$ such that $S$ is connected, $S$ contains $v_0$ and the ratio $\frac{|N[S]|}{|S|}$ is maximized, where $N[S]$ denotes the \emph{closed neighborhood} of $S$, that is, $N[S]$ contains all nodes in $S$ and all nodes with at least one neighbor in $S$.
		
		We prove that the problem is NP-hard even when the input graph $G$ is restricted to be a split graph. 
		On the positive side, we demonstrate a polynomial time approximation scheme for split graphs.
		Furthermore, we present a $\frac{1}{6}(1-\frac{1}{e})$-approximation algorithm for general graphs based on techniques for the \emph{Budgeted Connected Domination} problem [Khuller et al.,\ \emph{SODA 2014}].
		Finally, we provide a polynomial-time algorithm for the special case of interval graphs.
		Our algorithm returns an optimal solution for MRCE in $\mathcal{O}(n^3)$ time, where $n$ is the number of nodes in $G$.
		
	\end{abstract}
	
	\section{Introduction} 
	
	In the evergrowing World Wide Web landscape, browsers compete against each other to offer the best quality of surfing to their users.
	A key characteristic in terms of quality is the speed attained when retrieving a new page or, in general, resource.
	Thus, a browser's objective is to minimize latency when moving from one resource to another.
	One way to achieve this goal is via \emph{prefetching}: when the user lies at a certain Web node, predict what links she is more likely to visit next and preload them in cache so that, when the user selects to visit one of them, the transition appears to be instantaneous.
	Indeed, the World Wide Web Consortium (W3C) provides standards for prefetching in HTML \cite{W3C}.
	Also, besides being nowadays a common practice for popular browsers, prefetching constitutes an intriguing research theme, e.g., see the surveys in \cite{Wang, Ali} for further references.
	
	However, prefetching may come with a high network load cost if employed at a large scale.
	In other words, there is a trade-off that needs to be highlighted: more prefetching may mean less speed and even delays.
	For this reason, it becomes essential to acquire knowledge about the maximum number of resources to be prefetched over any potential Web nodes a surfer may visit.
	In this respect, Fomin et al.\ \cite{Prefetching} define the \emph{Surveillance Game} as a model for worst-case prefetching.
	The game is played by two players, namely the \emph{surfer} and the \emph{marker}, on a (directed) graph $G$ representing (some view of) the Web graph.
	The surfer controls a token initially lying at a designated pre-marked start node $v_0$.
	In each round, the marker marks, i.e., prefetches, up to $k$ so-far unmarked nodes during her turn and then the surfer chooses to move her token at a neighboring node of its current position.
	Notice that, once marked, a node always remains marked thereafter.
	The surfer wins if she arrives at an unmarked node, otherwise the marker wins if she manages to mark the whole graph before such an event occurs. 
	In optimization terms, the quantity under consideration is the \emph{surveillance number}, denoted $sn(G, v_0)$ for a graph $G$ and a start (root) node $v_0$, which is the minimum number of marks the marker needs to use per round in order to ensure that a surfer walking on $G$ (starting from $v_0$) never reaches an unmarked node.
	
	A main observation regarding the above game is that the surfer follows some connected trajectory on the graph $G$.
	Let $S$ stand for the set of nodes included in this trajectory.
	The marker's objective is to ensure that all nodes in $S$ or in the neighborhood of $S$ get marked promptly.
	Let $N[S]$ stand for the \emph{closed neighborhood} of $S$, i.e., $N[S]$ includes all nodes in $S$ and all nodes with at least one neighbor in $S$.
	Fomin et al.\ prove (Theorem 20 \cite{Prefetching}) that, for any graph $G$ and root $v_0$, it holds $sn(G, v_0) \ge \max \lceil \frac{|N[S]|-1}{|S|}\rceil$, where the maximum is taken over all subsets $S$ that induce a connected subgraph of $G$ containing $v_0$. Moreover, equality holds in case $G$ is a tree.
	That is, a ratio of the form $|N[S]|/|S|$ (minus one and ceiling operator removed for clarity) provides a good lower bound and possibly in many occasions a good prediction on the prefetching load necessary to satisfy an impatient Web surfer.
	Hence, in this paper, we believe it is worth to independently study the problem of determining $\max \frac{|N[S]|}{|S|}$ where the maximum is taken over all subsets $S$ inducing a connected subgraph of $G$ containing $v_0$.
	We refer to this problem as the \emph{Maximum Rooted Connected Expansion} problem (shortly MRCE) since we seek to find a connected set $S$ (containing the root $v_0$) maximizing its \emph{expansion ratio} in the form of $|N[S]|/|S|$.
	
	Except for the prefetching motivation, such a problem can stand alone as an extension to the well-studied family of domination problems. 
	Indeed, we later use connections between our problem and a domination variant in \cite{BCDS} to prove certain results.
	Finally, notice that removing the root node requirement makes the problem trivial.
	Let $\Delta$ stand for the maximum degree of a given graph $G$.
	Then, a solution consisting of a single max-degree node gives a ratio of $\Delta+1$.
	In addition, the ratio is at most $\Delta+1$, since given any connected set $S$ consisting of $k$ nodes, $|N[S]| \le (\Delta + 1)k$ due to the fact that each node can contribute at most $\Delta+1$ new neighbors (including itself).
	
	\subparagraph{Related Work.}
	The Surveillance Game was introduced in \cite{Prefetching}, where it was shown that computing $sn(G, v_0)$ is NP-hard in split graphs, nonetheless, it can be computed in polynomial time in trees and interval graphs. 
	Furthermore, in the case of trees, the MRCE ratio is proved \cite{Prefetching} to be equal to $sn(G, v_0)$ and therefore can be computed in polynomial time.
	In \cite{Connected}, the connected variant of the problem is considered, i.e., when the set of marked nodes is required to be connected after each round.
	For the corresponding optimization objective, namely the \emph{connected surveillance number} denoted $csn(G, v_0)$, it holds $csn(G,v_0) \le \sqrt{sn(G,v_0)n}$ for any $n$-node graph $G$.
	The more natural online version of the problem is also considered and (unfortunately) a competitive ratio of $\Omega(\Delta)$ is shown to be the best possible.
	
	A problem closely related to ours (as demonstrated later in Section~\ref{sec:general}) is the \emph{Budgeted Connected Dominating Set} problem (shortly BCDS), where, given a budget of $k$, one must choose a connected subset of $k$ nodes with a maximum size of closed neighborhood.
	This problem is shown to have a $(1-1/e)/13$-approximation algorithm (in general graphs) in \cite{BCDS}.
	
	Regarding problems dealing with some ratio of quantities, we are familiar with the \emph{isoperimetric number} problem \cite{Golovach}, where the objective is to \emph{minimize} $|\partial X|/|X|$ over all node-subsets $X$, where $\partial X$ denotes the set of edges with exactly one endpoint in $X$. 
	\emph{Vertex-isoperimetric} variants also exist; see for example \cite{Harper, Bezrukov}.
	Up to our knowledge, a ratio similar to the MRCE ratio we currently examine has not been considered.
	
	\subparagraph{Our Results.}
	We initiate the study for MRCE. 
	We prove that the decision version of MRCE is NP-complete, even when the given graph $G$ is restricted to be a split graph.
	For the same case, we demonstrate a polynomial-time approximation scheme running in $\mathcal{O}(n^{k+1})$ time with a constant-factor $\frac{k}{k+2}$ guarantee, for any fixed integer $k > 0$.
	Our algorithm exploits a growth property for MRCE and the special topology of split graphs.
	Moving on, we provide another algorithm for general graphs, i.e., when no assumption is made on the topology of the given graph besides it being connected. The algorithm is inspired by an approximation algorithm for BCDS \cite{BCDS} and achieves an approximation guarantee of $(1-1/e)/6$.
	Finally, we show that in the case of interval graphs, the MRCE ratio can be computed optimally in $\mathcal{O}(n^3)$ time for any given $n$-node graph.
	
	\subparagraph{Outline.}
	In Section~\ref{sec:prel}, we first define some necessary preliminary graph-theoretic notions and then formally define the MRCE problem.
	In Section~\ref{sec:split}, we present our results for split graphs.
	Later, in Section~\ref{sec:general}, we give the approximation algorithm for general graphs.
	Next, in Section~\ref{sec:interval}, we demonstrate the polynomial-time algorithm for interval graphs.
	Finally, in Section~\ref{sec:conclusion} we cite some concluding remarks and further work directions.
	
	\section{Preliminaries}\label{sec:prel}
	
	A graph $G$ is denoted as a pair $(V(G), E(G))$ of the nodes and edges of $G$.
	The graphs considered are simple (neither loops nor multi-edges are allowed), connected and undirected.
	
	Two nodes connected by an edge are called \emph{adjacent} or \emph{neighboring}.
	The \emph{open neighborhood} of a node $v \in V(G)$ is defined as $N(v) = \{u \in V(G): \{v, u\} \in E(G)\}$, while the \emph{closed neighborhood} is defined as $N[v] = \{v\} \cup N(v)$.
	For a subset of nodes $S \subseteq V(G)$, we expand the definitions of open and closed neighborhood as $N(S) = \bigcup_{v \in S} (N(v) \setminus S)$ and $N[S] = N(S) \cup S$.
	
	The degree of a node $v \in V(G)$ is defined as $d(v) = |N(v)|$.
	The minimum (resp.\ maximum) degree of $G$ is denoted by $\delta(G) = \min_{v \in V(G)} d(v)$ (resp.\ $\Delta(G) = \max_{v \in V(G)} d(v)$).
	
	A \emph{clique} is a set of nodes, where there exists an edge between each pair of them.
	The maximum size of a clique in $G$, i.e., the \emph{clique number} of $G$, is denoted by $\omega(G)$. 
	
	An \emph{independent set} is a set of nodes, where there exists no edge between any pair of them.
	The max.\ size of such a set in $G$, i.e., the \emph{independence number} of $G$, is denoted by $\alpha(G)$. 
	
	In the results to follow, we consider two specific families of graphs, namely \emph{split} and \emph{interval} graphs.
	Any necessary preliminary knowledge for these two graph families is given more formally in their corresponding sections.
	
	Finally, let us provide a formal definition of the quantity under consideration and the decision version of the corresponding optimization problem.
	
	\begin{definition}
		We define the Maximum Rooted Connected Expansion number for a graph $G$ and a node $v_0$  as follows, where $Con(G, v_0) \eqdef \{S \subseteq V(G) \;|\; v_0 \in S \text{ and } S \text{ is connected}\}$:
		$$\emph{MRCE}(G, v_0) = \max_{S\in Con(G, v_0)} \frac{|N[S]|}{|S|}$$
		
	\end{definition}
	
	\begin{definition}[$\mathcal{MRCE}$]
		Given a graph $G$, a node $v_0 \in V(G)$ and two natural numbers $a, b$, decide whether $\emph{MRCE}(G, v_0) \ge a/b$.
	\end{definition}
	
	When the input graph is known to be split, respectively interval, we refer to the corresponding optimization problem as $Split \MRCE$, respectively $Interval \MRCE$. 
	
	\section{Split Graphs}\label{sec:split}
	In this section, we define split graphs and cite a useful preliminary result regarding their structure.
	We proceed with our results and prove that $Split \MRCE$ is NP-hard, but it can be approximated within a constant factor of $\frac{k}{k+2}$ for any fixed integer $k > 0$.

	\begin{definition}
		A graph is split if it can be partitioned into a clique and an independent set.
	\end{definition}	
	
	Given the above definition, we denote by $(I,C)$ a partition for a split graph $G$ where $I$ stands for the independent set and $C$ for the clique.
	However, there may be many different ways to partition a split graph into an independent set and a clique \cite{Golumbic}.
	
	\begin{theorem}[Follows from Theorem 3.1 \cite{Feder}]\label{thm:partitions}
		A split graph has at most a polynomial number of partitions into a clique and an independent set.
		Furthermore, all these partitions can be found in polynomial time.
	\end{theorem}
	
	\subsection{Hardness}
	
	We now move onward to investigate the complexity of $Split \MRCE$.
	Initially, let us a define a pair of satisfiability problems we rely on in order to prove NP-hardness.
	
	\begin{definition}[$3$-$\mathcal{SAT}$]
		Given a CNF formula $\phi$ with $n$ variables and $m$ clauses, where each clause is a disjunction of exactly 3 literals, decide whether $\phi$ is satisfiable.
	\end{definition}
	
	\begin{definition}[$3$-$\mathcal{SAT}_{equal}$]
		Given a CNF formula $\phi$ with $n$ variables and $n$ clauses, where each clause is a disjunction of exactly 3 literals, decide whether $\phi$ is satisfiable.
	\end{definition}
	
	To demonstrate the hardness result in a more presentable way, we employ an auxiliary reduction from $3$-$\mathcal{SAT}$ to $3$-$\mathcal{SAT}_{equal}$ and then a reduction from $3$-$\mathcal{SAT}_{equal}$ to $Split \MRCE$.
	
	We recall that $3$-$\mathcal{SAT}$ is well-known to be NP-hard, e.g. see \cite{Garey}.
	\begin{lemma}\label{lem:3-sat-eq}
		$3$-$\mathcal{SAT}_{equal}$ is NP-hard.
	\end{lemma}
	
	\subparagraph{The Reduction.}
	Given a $3$-$\mathcal{SAT}_{equal}$ formula $\phi$, we create a graph $G$ with a node $v_0 \in V(G)$.
	Let $x_1, x_2, \ldots, x_n$ stand for the variables of $\phi$ and $c_1, c_2, \ldots, c_n$ for the clauses of $\phi$.
	We construct the graph $G$ in the following way:
	we place a node $v_0$, one node per literal $x_i, \overline{x_i}$ ($2n$ nodes in total), 
	one node per clause $c_i$ ($n$ nodes in total) and a set of $3n + 2$ "leaf" nodes for each variable (namely $y_{ij}$ for $j = 1, \ldots, 3n+2$) summing up to $(3n+2) \cdot n = 3n^2 + 2n$ "leaf" nodes in total.
	We call the two nodes $x_i, \overline{x_i}$ a \emph{literal-pair} and each node $c_i$ a \emph{clause-node}.
	Then, we connect $v_0$ to each literal node and each literal node to \emph{all} the other literal nodes. Moreover, each literal-node is connected to all the corresponding clause-nodes where it appears in $\phi$. 
	Finally, $x_i$ and $\overline{x_i}$ are connected to $y_{ij}$ for all $j$.
	It is clear that the construction can be done in polynomial time.
	Formally, $V(G) = \{v_0\}  \cup \{x_i, \overline{x_i}: 1 \le i \le n\} \cup \{c_{i}: 1 \le i \le n\} \cup \{y_{ij}: 1 \le i \le n, 1 \le j \le 3n+2 \}$
	and
	\begin{equation*}
	\resizebox*{0.95\linewidth}{!}{$	
		\begin{split}
		E(G) = &\; \{[v_0, x_i]: 1 \le i \le n\} \cup \{[v_0, \overline{x_i}]: 1 \le i \le n\} \cup\\
		&\cup \{[x_i, x_j]: 1 \le i, j \le n, i \neq j\} \cup \{[\overline{x_i}, x_j]: 1 \le i, j \le n, i \neq j\} \cup \{[\overline{x_i}, \overline{x_j}]: 1 \le i, j \le n, i \neq j\} \cup \\
		&\cup \{[x_i, y_{ij}]: 1 \le i \le n, 1 \le j \le 3n+2\} \cup \{[\overline{x_i}, y_{ij}]: 1 \le i \le n, 1 \le j \le 3n+2\} \cup \\
		&\cup \{[x_i, c_{j}]: x_i \text{ in clause } c_{j}\}
		\end{split}
		$}
	\end{equation*}
	That is, we get $|V(G)| = 1 + 5n + 3n^2$ and $|E(G)| = 2n + \binom{2n}{2} + 2n(3n+2) + 3n = 8n^2 + 8n$.
	Figure \ref{fig:red} demonstrates an example of such a construction;
	the literal-nodes within the dashed ellipsis form a clique.
	
	\begin{figure}
		\centering
		\includegraphics{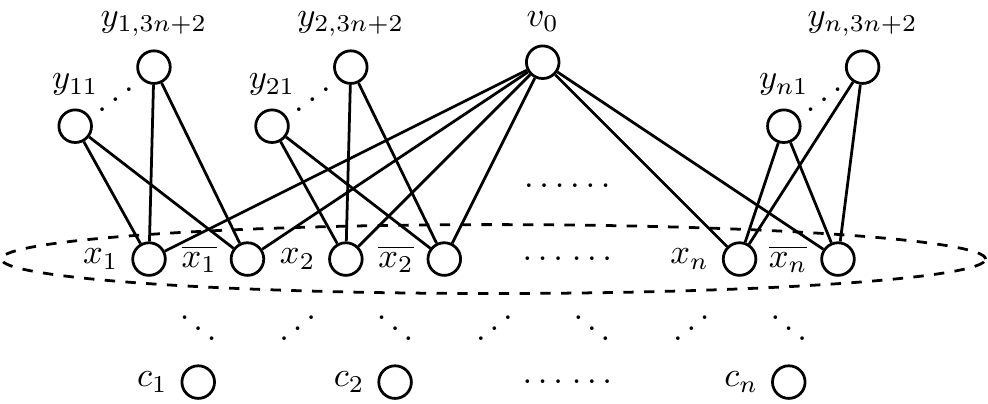}
		\caption{The graph $G$ constructed for the reduction}
		\label{fig:red} 
	\end{figure}
	
	\begin{proposition}
		$G$ is a split graph.
	\end{proposition}
	\begin{proof}
		$x_1, \overline{x_1}, x_2, \overline{x_2}, \ldots, x_n, \overline{x_n}$ form a clique; all other nodes form an independent set.
	\end{proof}
	
	\begin{claim}\label{claim:=>}
		If $\phi$ is satisfiable, then $\emph{MRCE}(G, v_0) \ge \frac{1+5n+3n^2}{1+n}$ .
	\end{claim}
	
	\begin{proof}
		Let $A$ stand for a truth assignment under which $\phi$ is satisfiable.
		Then, to form a feasible solution for MRCE, we choose a set $S$ including $v_0$ and these literal-nodes (either $x_i$ or $\overline{x_i}$) whose corresponding literals are set true under $A$. 
		Therefore, we get $|S| = 1+n$.
		Since, in $\phi$, each clause is satisfied by at least one literal set true under $A$, each clause-node $c_i$ is connected to at least one literal-node in $S$.
		Moreover, any node $y_{ij}$ is connected to $S$, since exactly one out of $x_i$ and $\overline{x_i}$ is in $S$ (due to $A$ being a truth assignment).
		Overall, we see that $|N[S]| = |V(G)| = 1 + 5n + 3n^2$.
	\end{proof}
	
	\begin{claim}\label{claim:<=}
		If there exists no satisfiable assignment for $\phi$, then $MRCE(G, v_0) < \frac{1+5n+3n^2}{1+n}$.
	\end{claim}
	
	\begin{proof}
		Let us first show a proposition to restrict the shape of a feasible MRCE solution.
		Intuitively, adding any $y_{ij}$ or $c_i$ node does not contribute any new neighbors to the ratio.
		\begin{proposition}\label{prop:reduction}
			Adding any $y_{ij}, c_i$ node can only decrease the ratio of a feasible solution.
		\end{proposition}
		
		The above proposition suggests it suffices to upper-bound potential solutions $S$ containing $v_0$ and only literal nodes. 
		Below, let $R = \frac{1+5n+3n^2}{1+n}$.
		To conclude the proof, we show that, if $\phi$ is unsatisfiable, then the ratio we can obtain is strictly less than $R$.
		
		If $S = \{v_0\}$, then the ratio we get is $\frac{|N[\{v_0\}]|}{|\{v_0\}|} = \frac{1 + 2n}{1} < R$ for any $n > 0$.
		
		If $S$ contains $v_0$ and $k$ literal nodes (any $k$ of them), we distinguish three cases.
		\begin{itemize}
			\item \emph{Case $k \le n-1$:} For a fixed $k$, the ratio becomes at most $\frac{1 + 3n + k(3n+2)}{1+k}$, since at most $k$ families of $y$ nodes are in the neighborhood.
			We observe $\partial\left(\frac{1 + 3n + k(3n+2)}{1+k}\right) /\partial k = \frac{1}{(k+1)^2} > 0$ for any $k > 0$.
			Hence, the worst case is $k = n-1$, which yields a ratio $\frac{1 + 3n + (n-1)\cdot(3n+2)}{n} = \frac{3n^2+2n-1}{n} < R$ for any $n > 0$.
			\item \emph{Case $k = n$:} If exactly one node from each literal pair is in $S$ (i.e. $S$ corresponds to a truth assignment), then the ratio becomes at most $\frac{1+3n-1+n(3n+2)}{1+n} < R$, since $\phi$ is unsatisfiable and therefore any truth assignment leaves at least one uncovered clause node.
			On the other hand, if there exists at least one literal-pair where both $x_i$ and $\overline{x_i}$ are not in $S$, then the ratio is at most $\frac{1+3n+(n-1)(3n+2)}{1+n} < R$, since at least one set of $3n+2$ "leaf" nodes are not in $N[S]$.
			\item \emph{Case $k > n$:} The ratio becomes at most $\frac{|(V(G)|}{1+k} = \frac{1+5n+3n^2}{1+k} < \frac{1+5n+3n^2}{1+n} = R$.
		\end{itemize}
	\end{proof}
	
	\begin{theorem}\label{thm:reduction}
		$Split \; \mathcal{MRCE}$ is NP-complete.
	\end{theorem}
	\begin{proof}
		By Claims~\ref{claim:=>} and \ref{claim:<=}, $Split \; \mathcal{MRCE}$ is NP-hard.
		$Split \; \mathcal{MRCE}$ is in NP, since given a potential solution $S \subseteq V(G)$, we can check in polynomial time whether $S$ is connected, $v_0 \in S$ and $|N[S]|/|S|$ satisfies the requested ratio.
	\end{proof}
	
	\subsection{Approximation} 
	
	We now turn our attention to a polynomial time approximation scheme for $Split \MRCE$.
	Our algorithm is parameterized by any fixed integer $k > 0$ and provides an approximation guarantee of $\frac{k}{k+2}$.
	Intuitively, the idea is that, given the best MRCE ratio when the set size is restricted to be at most $k+2$, the overall optimal ratio cannot be much better due to a ratio growth property.
	Additionally, connectivity is ensured due to the special topology of split graphs.
	Below, the approach is described formally in Algorithm~\ref{alg:split-improved}.
	Lemma~\ref{lem:ind} restricts the structure of a feasible MRCE solution on split graphs and the analysis follows in Theorem~\ref{thm:split-new}.
	
	\begin{algorithm}
		\SetKwInOut{Input}{Input}
		\SetKwInOut{Output}{Output}
		\DontPrintSemicolon
		
		\Input{A split graph $G = (V(G), E(G))$, a node $v_0 \in V(G)$ and a fixed integer $k > 0$}
		\Output{An MRCE solution and its corresponding ratio as a pair}
		
		\vspace{0.5mm}
		
		$S_{apx} \leftarrow \argmax_{S \in Con(G, v_0), 1 \le |S| \le k+2} |N[S]|/|S|$\;
		
		\vspace{0.5mm}
		
		\textbf{return} $(S_{apx}, |N[S_{apx}]|/|S_{apx}|  )$\;
		
		\caption{Approximate Split MRCE}
		\label{alg:split-improved}
	\end{algorithm}
	
	\begin{lemma}\label{lem:ind}
		Let $G$ be a split graph, $v_0 \in V(G)$ the requested root node and $(I,C)$ a partition of $G$ into an independent set $I$ and a clique $C$ where $|C| = \omega(G)$.
		Any feasible solution for $Split \MRCE$ containing nodes in $I$ can be transformed into another feasible solution with no nodes in $I$ (except maybe for $v_0$) which achieves a non-decreased MRCE ratio. 
	\end{lemma}
	\begin{proof}
		Suppose we are given a set $S\in Con(G, v_0)$, where $S \cap I \neq \emptyset$.
		We form a new feasible solution $S'$ as follows: include $v_0$ and all nodes in $S \cap C$. 
		Then, for each node $u \in (S \setminus \{v_0\}) \cap I$, let $u' \in N(u)$ stand for an arbitrarily selected neighbor of $u$.
		If $u' \notin S'$, add $u'$ to $S'$, otherwise proceed.
		Notice that $u' \in C$ since $u \in I$ and so $N(u) \subseteq C$.
		Thus, in the end it holds $(S' \setminus \{v_0\}) \cap I = \emptyset$.
		
		Now, let us compare the MRCE ratios of the two solutions.
		By construction, we know $|S'| \le |S|$ since the clique nodes of $S$ are surely in $S'$ and some more nodes may be added but at most as many as the independent set nodes of $S$.
		Moreover, it holds $|N[S']| \ge |N[S]|$, since for each pair $u, u'$ mentioned above we get $N[u] \subseteq N[u']$.
		That is, $u'$ contributes at least as many neighbors as $u$, i.e., $N(u) \subseteq N(u')$, since $u' \in C$ means $C \subseteq N(u')$ and $u \in I$ implies $N(u) \subseteq C$.
		Overall, we get $|N[S']|/|S'| \ge |N[S]|/|S|$.
	\end{proof}
	
	\begin{theorem}\label{thm:split-new}
		For any fixed integer $k > 0$, Algorithm~\ref{alg:split-improved} runs in $\mathcal{O}(n^{k+1})$ time and returns a $\frac{k}{k+2}$-approximation for $Split \MRCE$.
	\end{theorem}
	\begin{proof}
		The algorithm computes a maximum value out of all connected subsets of size at most $k+2$, including $v_0$, and so it runs in $\mathcal{O}(n^{k+1})$ time.
		
		Let $S_{opt}$ stand for an optimal solution for $Split \MRCE$. 
		In other words, it holds $S_{opt} \in \argmax_{S \in Con(G, v_0)} |N[S]|/|S|$.
		We distinguish two cases based on the size of $S_{opt}$.
		
		If $|S_{opt}| \le k+2$, then Algorithm~\ref{alg:split-improved} considers $S_{opt}$ and either returns it or another solution achieving the same ratio.
		
		If $|S_{opt}| > k+2$, then consider the following procedure: repeatedly remove from $S_{opt}$ the node with the least contribution in the numerator until $k$ nodes are left.
		More formally, let us denote $|S_{opt}| = l$ and then $S_{opt} = S_l$.
		For $i = l-1, \ldots, k$, let $S_{i} = S_{i+1} \setminus \{u_{i+1}\}$ for some node $u_{i+1}$ that maximizes $|N[S_{i+1} \setminus\{v\}]|$ over all $v \in S_{i+1}$.
		Equivalently, let $p(v) = |N[S_{i+1}]| - |N[S_{i+1} \setminus \{v\}]|$ denote the number of exclusive neighbors of $v$ in $N[S_{i+1}]$.
		Then, $u_{i+1} \in \argmin_{v\in S_{i+1}} p(v)$. 
		Notice that, for any $i = l-1, \ldots, k$, it may be the case that $S_i$ is \emph{not} a feasible MRCE solution, since $v_0$ may be removed during this process.
		
		Now, let us show that the ratio does not decrease while performing the above process.
		For any $i \in \{l-1, \ldots, k\}$, let $|N[S_i]| = N_i$ and $|S_i| = n_i$.
		Assume $\frac{N_{i+1}}{n_{i+1}} > \frac{N_{i}}{n_{i}}$.
		We rewrite the inequality as $\frac{N_{i+1}}{n_{i+1}} > \frac{N_{i+1} - p(u_{i+1})}{n_{i+1}-1}$ which implies $p(u_{i+1}) > \frac{N_{i+1}}{n_{i+1}}$. 
		Since $u_{i+1}$ minimizes the value of $p(\cdot)$, it follows that, for every $v \in S_{i+1}$, $p(v) \ge p(u_{i+1})$.
		Furthermore, $N_{i+1} \ge \sum_{v \in S_{i+1}} p(v)$ because $N[S_{i+1}]$ includes all exclusive neighbors of each node.
		Putting everything together, we get $N_{i+1} \ge \sum_{v \in S_{i+1}} p(v) > \sum_{v \in S_{i+1}} \frac{N_{i+1}}{n_{i+1}} = n_{i+1} \frac{N_{i+1}}{n_{i+1}} = N_{i+1}$, a contradiction.
		Based on this observation, we get $\frac{N_k}{n_k} \ge \frac{N_{k+1}}{n_{k+1}} \ge \ldots \ge \frac{N_l}{n_l} = OPT$, where $OPT$ stands for the optimal MRCE number.
	
		From Lemma~\ref{lem:ind}, we may assume without loss of generality that $S_{opt}\setminus\{v_0\} \subseteq C$.
		Moreover, due to the removal procedure followed, $S_k\setminus\{v_0\} \subseteq S_{opt}\setminus\{v_0\} \subseteq C$.
		In the worst case, when $v_0 \in I$ and $v_0$ has no neighbor in $S_{k}$, we form $S' = S_k \cup \{v_0, r\}$ where $r \in N(v_0)$ is a representative of $v_0$ in the clique $C$ such that $S_k \subseteq N(r)$.
		Notice that, since $S' \supseteq S_k$, then $N[S'] \supseteq N[S_k]$.
		Since $|S'|=k+2$, $S'$ is considered by Algorithm~\ref{alg:split-improved} and therefore it holds $\frac{|N[S_{apx}]|}{|S_{apx}|} \ge \frac{|N[S']|}{|S'|}$ where $S_{apx}$ is the solution returned by Algorithm~\ref{alg:split-improved}.
		Overall, we get the approximation guarantee $\frac{|N[S_{apx}]|}{|S_{apx}|} \ge \frac{|N[S']|}{|S'|} \ge \frac{|N[S_k]|}{k + 2} = \frac{k}{k+2} \frac{|N[S_k]}{k} \ge \frac{k}{k+2}\frac{N_l}{n_l} = \frac{k}{k+2}OPT$.	
	\end{proof}
	
	\section{General Graphs}\label{sec:general}
	
	We hereby state a constant-factor approximation algorithm for the general case when the input graph $G$ has no specified structure. 
	Our algorithm and analysis closely follow the work in \cite{BCDS} for the \emph{Budgeted Connected Dominating Set} (shortly BCDS) problem.
	
	In BCDS, the input is a graph $G$ with $n$ vertices and a natural number $k$ and we are asked to return a \emph{connected} subgraph, say $S$, of at most $k$ vertices of $G$ which maximizes the number of dominated vertices $|N[S]|$. 
	Khuller et al.\ \cite{BCDS} prove that there is a $(1-1/e)/13$ approximation algorithm for BCDS.
	In broad lines, their algorithmic idea is to compute a greedy dominating set and its corresponding \emph{profit function} and then obtain a connected subgraph via an approximation algorithm for the \emph{Quota Steiner Tree} (shortly QST) problem.
	
	\begin{definition}[$\mathcal{QST}$]
		Given a graph $G$, a node profit function $p: V(G) \rightarrow \mathbb{N} \cup \{0\}$, an edge cost function $c: E(G) \rightarrow \mathbb{N} \cup \{0\}$ and a quota $q \in \mathbb{N}$, find a subtree $T$ that minimizes $\sum_{e\in E(T)} c(e)$ subject to the condition $\sum_{v \in V(T)} p(v) \ge q$. 
	\end{definition}
	
	Evidently, both MRCE and BCDS require finding a connected subset $S \subseteq V(G)$ with many neighbors. 
	Nonetheless, while in BCDS we only care about maximizing $|N[S]|$, in MRCE we care about maximizing $|N[S]|/|S|$ with the additional demand that $v_0 \in S$.
	In order to deal with this extra requirement, in this paper, we are going to employ the rooted version of QST, namely the \emph{Rooted Quota Steiner Tree} (shortly RQST) problem.
	
	\begin{definition}[$\mathcal{RQST}$]
		Given a graph $G$, a root $v_0 \in V(G)$, a profit function $p: V(G) \rightarrow \mathbb{N} \cup \{0\}$, an edge cost function $c: E(G) \rightarrow \mathbb{N} \cup \{0\}$ and a quota $q \in \mathbb{N}$, find a subtree $T$ that minimizes $\sum_{e\in E(T)} c(e)$ subject to the conditions $\sum_{v \in V(T)} p(v) \ge q$ and $v_0 \in T$. 
	\end{definition}
	
	Garg \cite{Garg} gave a $2$-approximation algorithm for the (rooted) $k$-\emph{Minimum Spanning Tree} (shortly $k$-MST) problem based on the Goemans-Williamson \emph{Prize-Collecting Steiner Tree} approximation algorithm (shortly GW) \cite{G1, G2}. 
	Johnson et al.\ \cite{QST} showed that any polynomial-time $\alpha$-approximation algorithm for (rooted) $k$-MST, which applies GW, yields a polynomial-time $\alpha$-approximation algorithm for (rooted) QST. 
	Hence, Theorem~\ref{thm:RQST} below follows.
	
	\begin{theorem}[\cite{Garg, QST}]\label{thm:RQST}
		There is a $2$-approximation algorithm for $RQST$.
	\end{theorem}

	\subparagraph{The Algorithm.}
	Algorithm~\ref{alg:greedy}, namely the \emph{Greedy Dominating Set} (shortly GDS) algorithm, describes a greedy procedure to obtain a dominating set and a corresponding profit function for the input graph $G$.
	At each step, a node dominating the maximum number of the currently undominated vertices is chosen for addition into the dominating set.
	
	Algorithm~\ref{alg:mrce}, namely the \emph{Greedy MRCE} algorithm, makes use of GDS to obtain a dominating set for a slightly modified version of $G$, namely a graph $G'$, which is the same as $G$ with the addition of $n^2$ leaves to node $v_0$.
	Then, the algorithm outputs a connected subset $T_i$ (containing $v_0$) for any possible size $i$.
	Finally, the subset yielding the best MRCE ratio is chosen as our approximate solution.
	
	In terms of notation, we refer to the approximation algorithm implied by Theorem~\ref{thm:RQST} as the $2$-$RQST(G, v_0, p, q)$ algorithm with a graph $G$, a root node $v_0 \in V(G)$, a profit function $p: V(G) \rightarrow \mathbb{N} \cup \{0\}$ and a quota $q$ as input. 
	We omit including an edge cost function, since in our case all edges have the same cost, that is, cost $1$. 
	Furthermore, let $[n] \eqdef \{1, 2, 3, \ldots, n\}$. 
	
	Now, consider a connected set $S_i$ of size $i$ (which contains $v_0$) yielding the maximum number of dominated vertices, i.e. $S_i \in \argmax_{S:\; S\in Con(G, v_0),\; |S| = i} |N[S]|$. 
	We then denote $OPT_i \eqdef |N[S_i]|$ and use it in the quota parameter of $2$-$RQST$ at line $4$ of Greedy MRCE. 
	Yet, in the general case, we do not know $OPT_i$ and also such a quantity may be hard to compute. 
	To overcome this obstacle, notice that $OPT_i \in [i, n]$ and therefore we could \emph{guess} $OPT_i$, e.g., by running a sequential or binary search within the loop of Greedy MRCE and then keeping the best tree returned by $2$-$RQST$. 
	Notice that such an extra step requires at most a linear time overhead. Therefore, the running time of Greedy MRCE remains polynomial and is dominated by the running time of $2$-$RQST$.
	For presentation purposes, we omit this extra step and assume $OPT_i$ is known for each $i \in [n]$.
	
	In the analysis to follow, we focus on why this specific $(1-1/e)OPT_i$ quota is selected and how it leads to a $(1-1/e)/6$ approximation factor.
	
	\begin{figure}[h]
		\scalebox{0.9}{
			\begin{algorithm}[H]
				\SetKwInOut{Input}{Input}
				\SetKwInOut{Output}{Output}
				\DontPrintSemicolon
				
				\Input{A graph $G = (V(G), E(G))$}
				\Output{A dominating set $D \subseteq V(G)$ and a profit function $p: V(G) \rightarrow \mathbb{N} \cup \{0\}$}
				
				$D \leftarrow \emptyset$\;
				$U \leftarrow V(G)$\;
				\ForEach{$\upsilon \in V(G)$}{
					$p(\upsilon) \leftarrow 0$\;
				}
				\While{$U \neq \emptyset$}{
					$w \leftarrow \argmax_{\upsilon \in V(G) \setminus D} |N_U(\upsilon)|$ \tcc*{$N_U(\upsilon) = N[\{\upsilon\}]\cap U$}
					$p(w) \leftarrow |N_U(w)|$\;
					$U \leftarrow U \setminus N_U(w)$\;
					$D \leftarrow D \cup \{w\}$\;
				}
				\textbf{return} $(D, p)$
				
				\caption{Greedy Dominating Set (GDS) \cite{BCDS}}
				\label{alg:greedy}
			\end{algorithm}
		}
	\end{figure}
	\begin{figure}[h]
		\scalebox{0.9}{
			\begin{algorithm}[H]
				\SetKwInOut{Input}{Input}
				\SetKwInOut{Output}{Output}
				\DontPrintSemicolon
				
				\Input{A graph plus node pair $(G, v_0)$}
				\Output{An MRCE solution $S$ and its corresponding ratio $s$}
				
				Construct $G'$: same as $G$ with extra $n^2$ leaves attached to $v_0$ \\
				$(D, p) \leftarrow GDS(G')$\;
				\ForEach{$i \in [n]$}{
					$T_i \leftarrow $ $2$-$RQST(G, v_0, p, (1-\frac{1}{e})OPT_i)$ \\
				}
				Let $i^* = \argmax_{i \in [n]} |N[T_i]|/|T_i|$\;
				\textbf{return} $(T_{i^*}, |N[T_{i^*}]|/|T_{i^*}|)$\;
				
				\caption{Greedy MRCE}
				\label{alg:mrce}
			\end{algorithm}
		}
	\end{figure}
	
	\subparagraph{Analysis. }
	Let us consider some step $i$ of the loop in the Greedy MRCE algorithm.
	Recall that $OPT_i = \max_{S:\; S\in Con(G, v_0),\; |S| = i} |N[S]|$. That is, $OPT_i$ stands for the maximum number of dominated vertices by a connected subset of size $i$, which contains $v_0$. 
	In the call to $2$-$RQST$, notice that, although $OPT_i$ refers to the graph $G$ and by definition contains $v_0$, the profit function $p$ (as well as the corresponding greedy dominating set $D$) stems from running GDS on $G'$.
	The reason for this choice is, due to the extra $n^2$ leaves attached to $v_0$ in $G'$, to force $v_0$ into the greedy dominating set $D$ and assign to it the highest profit amongst all nodes.
	Below, let $S_{i,G'} \in \argmax_{S:\; S\subseteq V(G),\; |S| = i,\; S\text{ is connected}} |N[S]|$ and $OPT_i^{G'} \eqdef |N[S_{i, G'}]|$, i.e., $OPT_{i,G'}$ denotes the maximum number of nodes dominated by a size-$i$ subset of nodes in $G'$.
	
	\begin{claim}\label{claim:SiG}
		For any $i \in [n]$, it holds $v_0 \in S_{i, G'}$ .
	\end{claim}
	\begin{proof}
		Suppose $v_0 \notin S_{i, G'}$ for some $i \in [n]$.
		$S_{i, G'}$ consists of $i$ vertices each contributing at most $\Delta(G)$ neighbors in terms of domination.
		Thence, $OPT_{i, G'} \le i + i\cdot\Delta(G) = i(\Delta(G)+1) \le n^2$ since $i \le n$ and $\Delta(G) \le n-1$.
		However, we can pick another subset including $v_0$ and $i-1$ leaves of $v_0$ to get at least $n^2 + 1$ dominated nodes, i.e., $v_0$ and all its leaves. 
	\end{proof}
	
	Let us introduce some further notation for the proofs to follow.
	Let $L_1 = S_{i, G'}$ and $L_2 = N(L_1)$, that is, $OPT_{i,G'} = |L_1 \cup L_2|$.
	Also, let $L_3 = N(L_2)\setminus L_1$ and $R = V(G) \setminus (L_1 \cup L_2 \cup L_3)$, where $R$ denotes the remaining vertices, i.e., those outside the three layers $L_1, L_2, L_3$.
	Let us now consider the intersection of these layers with the greedy dominating set $D$ returned by GDS.
	Let $L'_j = D \cap L_j$ for $j = 1,2,3$ and $D'_i = \{v_1, v_2, \ldots, v_i\}$ denote the first $i$ vertices from $L'_1 \cup L'_2 \cup L'_3$ in the order selected by the greedy algorithm. 
	In order to bound the total profit in $D'_i$, we define $g_j = \sum_{k=1}^j p(v_k)$ as the profit we gain from the first $j$ vertices of $D'_i$.
	
	\begin{claim}[Variation of Claim 1 in \cite{BCDS}]\label{claim:g_j}
		It holds $g_{j+1} - g_{j} \ge \frac{1}{i}(OPT_{i, G'} - g_j)$.
	\end{claim}
	\begin{proof}
		Consider the iteration of GDS where $v_{j+1}$ is picked for inclusion in $D$.
		Any node $w \in L_1 \cup L_2$, which is already dominated by some node in $D$, must be dominated by a node of $D'_i$ in $\{v_1, \ldots, v_j\}$, since $w$ cannot be dominated by a node lying in $R$.
		Hence, at most $g_j$ vertices of $L_1 \cup L_2$ are dominated thus far.
		Equivalently, at least $|L_1 \cup L_2| - g_j = OPT_{i, G'} - g_j$ vertices remain undominated.
		Since $|L_1| = i$ vertices neighbor all the above undominated ones, by a pigeonhole argument,
		there exists at least one node $u \in L_1$ (and $u \notin D$) which neighbors at least $\frac{1}{i}(OPT_{i, G'} - g_j)$ of them.
		Since GDS picked $v_{j+1}$ at this iteration instead of $u$, it follows $p(v_{j+1}) \ge p(u) \ge \frac{1}{i}(OPT_{i, G'} - g_j)$, where $p(v_{j+1}) = g_{j+1} - g_j$.
	\end{proof}
	
	\begin{lemma}[Variation of Lemma 5.1 in \cite{BCDS}]\label{lemma:D_i}
		There exists a subset $D'_i \subseteq D$ of size $i$ with total profit at least $(1-\frac{1}{e})OPT_i$.
		Further, $D'_i$ can be connected using at most $2i$ Steiner nodes and contains $v_0$.
	\end{lemma}
	\begin{proof}
		By solving the recurrence from Claim \ref{claim:g_j}, we get $g_j \ge (1 - (1-\frac{1}{i})^j)OPT_{i, G'}$.
		Thence,
		$$ \sum_{v \in D'_i} p(v) = g_i \ge \left(1 - \left(1-\frac{1}{i}\right)^i\right)OPT_{i, G'} \ge \left(1- \frac{1}{e}\right)OPT_{i, G'} \ge \left(1- \frac{1}{e}\right)OPT_i$$
		
		since $(1-\frac{1}{i})^i \le 1/e$ for $i \ge 1$ and $OPT_{i, G'} \ge OPT_i + n^2$, since the subset $S_i$, where $N[S_i] = OPT_i$, is a feasible solution for the maximum number of dominated vertices in $G'$, giving a number equal to $OPT_i$ plus the $n^2$ $v_0$-leaves present in $G'$.
		
		Now, let us show that an extra $2i$ nodes are enough to ensure that $D'_i$ is connected.
		We select a subset $D''_i \subseteq L_2$  of size at most $|L_3 \cap D'_i| \le i$ to dominate all vertices of $D'_i \cap L_3$.
		Then, we ensure that all vertices are connected by simply adding all the $i$ vertices of $L_1$.
		Thus, $\hat{D}_i = D'_i \cup D''_i \cup L_1$ induces a connected subgraph that contains at most $3i$ vertices (one of them being $v_0$).
	\end{proof}

	\begin{theorem}
		There exists a $\frac{1}{6}(1-\frac{1}{e})$-approximation for MRCE in general graphs.
	\end{theorem}
	\begin{proof}
		For each $i \in [n]$,  by Lemma~\ref{lemma:D_i}, there exists a solution of at most $3i$ vertices with profit at least $(1-\frac{1}{e})OPT_i$.
		In Algorithm~\ref{alg:mrce}, we run $2$-$RQST$, therefore obtaining a, connected and including $v_0$, solution of at most $6i$ vertices with profit at least $(1-\frac{1}{e})OPT_i$. 
		Let $APX_i$ stand for the MRCE ratio of the approximate solution corresponding to $T_i$.
		Then $$APX_i \ge \frac{(1-\frac{1}{e})OPT_i}{6i} = \frac{1}{6}\left(1-\frac{1}{e}\right) \frac{OPT_i}{i}$$
		Now, let $OPT$ stand for the optimal ratio for MRCE.
		Then, $OPT = \max_{i\in[n]} \left\{\frac{OPT_i}{i}\right\}$.
		Let $i^*$ be the solution size returned by Algorithm~\ref{alg:mrce} and $i_0 = \argmax_{i\in[n]} \left\{\frac{OPT_i}{i}\right\}$.
		Then, $APX_{i^*} \ge APX_{i_0} \ge \frac{1}{6}\left(1-\frac{1}{e}\right) OPT$, which concludes the proof.
	\end{proof}
	
	\section{Interval Graphs}\label{sec:interval}
	In this section, we provide an optimal polynomial time algorithm for the special case of \emph{interval graphs}.
	We commence with some useful preliminaries and then provide the algorithm and its correctness.
	
	\subparagraph{Preliminaries. }
	
	All intervals considered in this section are defined on the real line, closed and non-trivial (i.e., not a single point).
	Their form is $[\alpha, \beta]$, where $\alpha < \beta$ and $\alpha, \beta \in \mathbb{R}$.
	
	\begin{definition}
		A graph is called interval if it is the intersection graph of a set of intervals on the real line.
	\end{definition}
	
	Following the above definition, each graph node corresponds to a specific interval and two nodes are connected with an edge if and only if their corresponding intervals overlap.
	
	\begin{definition}
		Given an interval graph $G$, a realization of $G$ (namely $I(G)$) is a set of intervals on the real line corresponding to $G$, where
		\begin{itemize}
			\item for each node $v \in V(G)$, the corresponding interval is given by $I(v) \in I(G)$, and
			\item for $v, u \in V(G)$, $I(v)$ intersects $I(u)$ if and only if $[v, u] \in E(G)$.
		\end{itemize}
	\end{definition}
	
	Notice that we can always derive a realization, where \emph{all interval ends are distinct}.
	Suppose that two intervals share a common end. 
	One need only extend one of them by $\epsilon > 0$ chosen small enough such that neighboring relationships are not altered.
	
	Below, we provide a definition caring for the relative position of two intervals with regards to each other.
	Building on that, we define a partition of $V(G)$ with respect to the position of the vertices' corresponding intervals apropos of the $v_0-$interval.
	
	\begin{definition}\label{def:intervals}
		Given two intervals $x = [x_l, x_r]$ and $y = [y_l, y_r]$ 
		, we denote the following:
		\begin{itemize}
			\item $x \sqsubset y$, i.e. $x$ is contained in $y$, when $x_l > y_l$ and $x_r < y_r$.
			\item $x \cap_L y$, i.e. $x$ intersects $y$ to the left, when $x_l < y_l$ and $y_l < x_r < y_r$.
			\item $x \cap_R y$, i.e. $x$ intersects $y$ to the right, when $x_r > y_r$ and $y_l < x_l < y_r$.
			\item $x \prec_L y$, i.e. $x$ is strictly to the left of $y$, when $x_r < y_l$.
			\item $x \succ_R y$, i.e. $x$ is strictly to the right of $y$, when $x_l > y_r$.
		\end{itemize}
	\end{definition}
	
	\begin{definition}
		We define the following sets:
		\begin{itemize}
			\item Let $C \coloneqq \{v \in V(G): I(v_0) \sqsubset I(v)\}$. 
			Notice that $v_0 \notin C$.
			\item Let $C' \coloneqq \{v \in V(G): I(v) \sqsubset I(v_0)\}$.
			Notice that $v_0 \notin C'$.
			\item Let $C_L \coloneqq \{v \in V(G): I(v) \cap_L I(v_0) \}$.
			\item Let $C_R \coloneqq \{v \in V(G): I(v) \cap_R I(v_0) \}$.
			\item Let $L \coloneqq \{v \in V(G): I(v) \prec_L I(v_0) \}$.
			\item Let $R \coloneqq \{v \in V(G): I(v) \succ_R I(v_0) \}$.
		\end{itemize}
	\end{definition}
	
	\begin{proposition}\label{prop:partition}
		$(L, C_L, C', C, \{v_0\}, C_R, R)$ forms a partition of $V(G)$.
	\end{proposition}
	\begin{proof}
		To see the union, one needs to spot that $V(G) = (V(G)\setminus N[v_0])\cup N[v_0]$, where $N[v_0] = \{v_0\} \cup C \cup C' \cup C_L \cup C_R$ and $ V(G)\setminus N[v_0] = L \cup R$.
		Disjointness follows from Definition \ref{def:intervals}. 
		For instance, should $C_L \cap C_R = \{v\} \neq \emptyset$, then $I(v)_l < I(v_0)_l$ and $I(v)_l > I(v_0)_l$, a contradiction.
	\end{proof}
	
	Let us proceed with some useful propositions regarding the form of an optimal solution.
	
	\begin{proposition}\label{prop:C'}
		The addition of any node $v \in C'$ to any feasible $Interval \MRCE$ set does not increase the solution ratio.
	\end{proposition}
	\begin{proof}
		Suppose we extend a feasible solution $S$ by forming another feasible solution $S' = S \cup \{v\}$, where $v \in C'$.
		Then, $N[S'] = N[S]$, since $v$ is a neighbor of $v_0$ and $v$ has, at the best case, the same neighbors as $v_0$.
		The new ratio becomes 
		$\frac{|N[S']|}{|S'|}  =
		\frac{|N[S]|}{|S|+1} < 
		\frac{|N[S]|}{|S|}$.
	\end{proof}
	
	Let us now show that we need only care about a specific subset of $C$, namely $C^*$, 
	defined as $C^* \coloneqq \{v \in C \mid \nexists~v' \in C: v \neq v' \land I(v) \sqsubset I(v')\}$.
	That is, we restrict ourselves to those vertices whose corresponding intervals contain $I(v_0)$, but are not contained in any other interval.
	In other words, we are only interested in the intervals that \emph{maximally} contain $I(v_0)$.
	
	\begin{proposition}\label{prop:C*}
		Any feasible $Interval \MRCE$ solution $S \subseteq V(G)$ containing a node $v \in C\setminus C^*$ can be transformed into another feasible solution $S'$, where $v \notin S'$, with at least the same ratio as $S$.
	\end{proposition}
	\begin{proof}
		Suppose we are given a feasible solution $S$ containing a node $v \in C \setminus C^*$.
		Then, by definition, there exists a node $v' \in C$ such that $v \neq v'$ and $I(v) \sqsubset I(v')$.
		Moreover, notice that $I(v) \sqsubset I(v')$ implies that $N[v] \subseteq N[v']$, since any interval intersecting $I(v)$ also intersects $I(v')$.
		We consider two cases.
		If $v' \in S$, then we form the feasible solution $S_1 = S \setminus \{v\}$.
		The new ratio is $
		\frac{|N[S_1]|}{|S_1|} =
		\frac{|N[S]|}{|S|-1} >
		\frac{|N[S]|}{|S|}
		$, since $|S_1| = |S| - 1$ and $N[S_1] = N[S]$ given that $v$ is a neighbor of $v_0$ and its neighbors are also covered by $v'$.
		Otherwise, if $v' \notin S$, we form the feasible solution $S_2 = (S \setminus\{v\})\cup v'$.
		The new ratio is $
		\frac{|N[S_2]|}{|S_2|}  \geq
		\frac{|N[S]|}{|S|}
		$, since $|S_2| = |S|$ and $|N[S_2]| \geq |N[S]|$ given that $N[v] \subseteq N[v']$.
	\end{proof}
	
	\subparagraph{The Algorithm. }
	The general idea of the algorithm is to start from the feasible solution $\{v_0\}$ and then consider a family of the best out of all possible expansions, while maintaining feasibility, either moving toward the left or the right in terms of the real line.
	The key in this approach is that the left and right part of the graph are dealt with \emph{independently} from each other.
	Of course, special care needs to be taken when other intervals contain $I(v_0)$.
	During this left/right subroutine, we save a series of possible expansion stop-nodes with maximal ratio.
	In the end, we conflate each left ratio with each right ratio and pick the combination providing the maximum one.
	The algorithm is given in Algorithm \ref{alg:interval} and the other routines follow in Algorithms~\ref{alg:expand}, \ref{alg:combine}.
	We hereby provide a short description for each function.
	\begin{itemize}
		
		\item \emph{Interval}: This is the main routine. The input is an interval graph $G$ and a starting node $v_0 \in V(G)$.
		The output is a solution set together with its corresponding ratio.
		Initially, the algorithm computes a realization $I(G)$, a partition of $V(G)$ and the \emph{core} set $C^*$ as defined in the preliminaries.
		Then, possible left and right expansions to $\{v_0\}$ are sought.
		These are combined to get a best solution for this case.
		Finally, these basic steps are repeated for each $c \in C^*$ and the best are kept in the \emph{Sols} pool.
		It then suffices to calculate the max out of the best candidate solutions.
		
		\item \emph{Expand}: This function is responsible for providing a set of possible expansions either left or right of a starting node.
		A direction, the starting node, the realization, the node partition and a counter are given as input.
		The counter serves to save different solutions in a vector, which is returned as output.
		Notice that the solution vector is \emph{static}, i.e. it can be accessed by any recursive call.
		The main step of the function is to select a node whose interval intersects the starting interval to the requested direction.
		At the same time, this interval needs to be the farthest away in this direction, i.e., its left/right endpoint needs to be smaller/greater to any other candidate's.
		The potential expansion is saved and the function is called recursively with the new node as a start point.
		The process continues till no further expansion can be made, i.e., the farthest interval  is reached.
		The returned vector does contain a no-expansion solution (case $count = 0$).
		
		\item \emph{Combine}: This function takes as input the potential left and right expansions.
		It then computes a ratio for each possible combination of left and right expansions and outputs the solution and ratio pair attaining the maximum ratio for the given starting node-set.
		
		\item \emph{MaxRatio}: This routine simply returns the maximum set-ratio pair out of a set of different such pairs.
		\item \emph{Ratio}: Simply returns the MRCE ratio for a given set.
		
	\end{itemize}
	
	
	\begin{figure}
		\scalebox{0.8475}{
			\begin{algorithm}[H]
				\SetKwInOut{Input}{Input}
				\SetKwInOut{Output}{Output}
				\DontPrintSemicolon
				
				\Input{An interval graph plus node pair $(G, v_0)$}
				\Output{A set-ratio pair $(S, s)$}
				
				$I \longleftarrow Realization(G)$\;
				$P \longleftarrow Partition(G, I)$\;
				$C^* \longleftarrow Core(C, I)$\;
				$L_{sols} \longleftarrow Expand(L, v_0, I, P, 0)$\;
				$R_{sols} \longleftarrow Expand(R, v_0, I, P, 0)$\;
				$Sols \longleftarrow Combine(\{v_0\}, L_{sols}, R_{sols}, G)$\;
				\ForEach{$c \in C^*$}{
					$L_{sols} \longleftarrow Expand(L, c, I, P, 0)$\;
					$R_{sols} \longleftarrow Expand(R, c, I, P, 0)$\;
					$Sols \longleftarrow Sols \cup \{Combine(\{v_0, c\}, L_{sols}, R_{sols}, G)\}$\;
				} 
				\textbf{return} $MaxRatio(Sols)$\;
				\caption{Interval}
				\label{alg:interval}
			\end{algorithm}
		}
	\end{figure}

	\begin{figure}
		\scalebox{0.88}{
			\begin{algorithm}[H]
				\SetKwInOut{Input}{Input}
				\SetKwInOut{Output}{Output}
				\DontPrintSemicolon
				
				\Input{A direction, node, realization, partition and counter $(D, v, I, P, count)$}
				\Output{A vector of sets of nodes $Sols$}
				
				\If{$count == 0$}{
					$Sols(count) \longleftarrow \{v\}$\;
				}
				Pick $v'$ such that $I(v')$ is the farthest interval on direction $D$ with $I(v') \cap_D I(v)$\;
				\uIf{$\nexists$ such a $v'$}{
					\textbf{return} $Sols$\;
				}
				\Else{
					$Sols(count+1) \longleftarrow Sols(count) \cup \{v'\}$\;
					\textbf{return} $Expand(D, v', I, P, count+1)$\;
				}
				
				\caption{Expand}
				\label{alg:expand}
			\end{algorithm}
		}
	\end{figure}
	
	\begin{figure}
		\scalebox{0.88}{
			\begin{algorithm}[H]
				\SetKwInOut{Input}{Input}
				\SetKwInOut{Output}{Output}
				\DontPrintSemicolon
				
				\Input{A node-set, left/right possible solutions and graph $(S, Left, Right, G)$}
				\Output{A set-ratio pair $(Argmax, Max)$}
				
				$(Argmax, Max) \longleftarrow (S, Ratio(S))$\;
				\ForEach{$l \in Left$}{
					\ForEach{$r \in Right$}{
						\If{$Ratio(S \cup l \cup r) > Max$}{
							$(Argmax, Max) \longleftarrow (S \cup l \cup r, Ratio(S \cup l \cup r))$\;
						}
					}
				}
				\textbf{return} $(Argmax, Max)$\;
				
				\caption{Combine}
				\label{alg:combine}
			\end{algorithm}
		}
	\end{figure}	
	
	
	\subparagraph{Correctness \& Complexity. }
	Lemma~\ref{lem:expand} argues about the fact that the solutions $Expand()$ ignores do not have any effect on optimality.
	We state the lemma for the \emph{left} expansion case and the reader can similarly adapt it to the right expansion case.
	Then, we conclude with the optimality and running time of the overall procedure (Theorem~\ref{thm:interval}).
	
	\begin{lemma}\label{lem:expand}
		Let $L_{sols}$ stand for the vector returned by the function call $Expand(L, v, I, P, 0)$ for some node $v \in V(G)$.
		For any node-set $S \subseteq C_L \cup L \cup \{v\}$ such that $v \in S$ and $S \notin L_{sols}$, there exists a set $S' \in L_{sols}$ such that $Ratio(S') \ge Ratio(S)$.
	\end{lemma}
	\begin{proof}
		Let $v = v_1, v_2, \ldots, v_k$ be the set of nodes picked in the recursive calls of $Expand()$ (in decreasing order of their right endpoint).
		Let $v = v'_1, v'_2 \ldots, v'_{k'}$ be the set of nodes in $S$ (again in decreasing order of their right endpoint).
		Since $S \notin L_{sols}$, there exists a node $v'_i \in S$ such that $v'_i \neq v_i$, i.e. a point where $S$ and $S'$ "diverge". 
		Then, we can replace $v'_i$ by $v_i$, since due to the choice of $v'_i$ in line $4$ of Algorithm~\ref{alg:expand} it holds $N_L(v'_i) \subseteq N_L(v_i)$, where $N_L(v)$ stands for the \emph{left} neighbors of $v$ (i.e. the neighbors whose corresponding intervals intersect $v$ to the left).
		Hence, after this replacement, the ratio of the set does not decrease due to the (possibly) increased size of the left neighborhood.
		Afterward, one can ignore all nodes $v'_j$ (where $j > i$) such that $I(v'_j) \sqsubset I(v_i)$ and repeat the same argument with $v_i$ as a starting point and so forth.
	\end{proof}
	
	\begin{theorem}\label{thm:interval}
		$Interval(G, v_0)$ optimally solves $Interval$ $\MRCE$ in $\mathcal{O}(n^3)$ time.
	\end{theorem}
	\begin{proof}
		For each node $v \in \{v_0\}\,\cup\, C^*$ that we choose as a starting point for the $Expand()$ function,
		we see that, when expanding with $v'$ such that $I(v') \cap_L I(v)$, $v'$ does not have any right-neighbors not already in $N_R(v)$.
		Equivalently, if we expand to the right, there is no effect on the left neighborhood of the starting node.
		Indeed, only intervals containing $v$ could harm this notion of left/right neighborhood \emph{independence} and these are not considered by $Expand()$.
		So, we can independently expand leftward and rightward and get a series of connected subsets in both directions.
		Then, $Combine()$ ensures we select the best left and right expansion in ratio terms by looking at all possible combinations. 
		Such a solution is actually a potential optimal: any subset ignored by $Expand()$ would yield a worse ratio (Lemma \ref{lem:expand}). Eventually, the maximum ratio amongst all possible starting points is returned.
		This is an overall optimal, since it outperforms all other potential optimals and we have considered all possible maximal intervals containing $v_0$, i.e., the set $C^*$, as part of the solution.
		
		$Realization()$ and $Partition()$ take linear time, while $Core()$ may take $ \mathcal{O}(n^2)$ time.
		The loop iterating the elements of $C^*$ in \emph{Interval()} dominates the time complexity.
		In the worst-case, $\mathcal{O}(n)$ steps for $Expand()$ and $\mathcal{O}(n^2)$ steps for $Combine()$ are repeated for $\mathcal{O}(n)$ elements of $C^*$.
		Thence, the worst-case time complexity is $\mathcal{O}(n^3)$.
	\end{proof}

\section{Conclusion \& Further Work}\label{sec:conclusion}

We proved that MRCE is NP-complete for split graphs. We showed that, in this case, the problem admits a polynomial time approximation scheme, whereas for interval graphs we proposed a polynomial-time algorithm. For general graphs, we also gave a constant-factor approximation algorithm by exploring the relation of MRCE with BCDS \cite{BCDS}.

The major open question is to improve the approximability of the problem on general graphs without applying BCDS techniques, but using rather MRCE properties. Another open problem is the design of an approximation algorithm for chordal graphs. Towards this direction, we notice that even for chordal graphs with a dominating clique (a superclass of split graphs), equivalently chordal graphs with diameter at most three (Theorem 2.1 \cite{Kratsch}), the assumption that only clique nodes need to be included in a solution (Lemma~\ref{lem:ind}) now fails.

\newpage %




\end{document}